\RequirePackage{fix-cm}
\documentclass[smallextended]{svjour3}     
\usepackage{graphicx}
 \usepackage{mathptmx}      
\usepackage{latexsym}
\usepackage{amssymb,amsfonts,amsmath}
\graphicspath{{Figures/}}

\begin{document}

\title{Turing patterns from dynamics of early HIV infection over a two-dimensional surface}
\author{O. Stancevic \and C. Angstmann \and J.M. Murray \and B.I. Henry}
\institute{O. Stancevic \and C. Angstmann \and B.I. Henry \at
          School of Mathematics and Statistics, University of New South Wales, Sydney NSW 2052, Australia.\\
          \email{b.henry@unsw.edu.au}
         \and
         J.M. Murray \at
           School of Mathematics and Statistics, University of New South Wales, Sydney NSW 2052, Australia.\\
          The Kirby Institute, University of New South Wales,  Sydney NSW 2052, Australia.\\
         }

\date{Version date \today} 
\maketitle

\begin{abstract}
We have developed a mathematical model for in-host virus dynamics that includes
spatial chemotaxis and diffusion across a two dimensional surface representing the vaginal or rectal epithelium at primary HIV infection.
A linear stability analysis of the steady state solutions identified conditions for Turing instability pattern formation.
We have solved the model equations numerically using parameter values obtained from previous experimental results for HIV infections.
Simulations of the model for this surface show hot spots of infection.
Understanding this localization is an important step in the ability to correctly model early HIV infection. These spatial variations also have implications for the development and effectiveness of microbicides against HIV

\keywords{HIV \and Turing patterns \and In-host viral dynamics \and Chemotaxis \and Reaction-diffusion}
\subclass{92C15 \and 92C17 \and 92C50 \and 35B36 \and 35K57 \and 92D30 \and 60J70}
\end{abstract}
\section{Introduction}
\label{sect:intro}

The normal response in an individual after being infected by a virus is activation of the immune system, driving infection levels down. If the immune response is sufficiently potent then the disease can be completely eradicated from the body, but in many instances this does not occur. Instead, over the course of time, an eventual balance of disease replication and immune clearance is established leading to chronic infection. These steady state outcomes, clearance versus chronic infection, are suggestive of simple dynamics but this ignores  spatial variations, including possible hot spots of infection \cite{Haase1996},  which confound the dynamics in both transient and long time behaviours.
Although there has been some mathematical modelling of acute HIV infection \cite{Murray1998,DeBoer2007,Ribeiro2010}, current mathematical models of HIV infection mainly focus on response  to antiretroviral therapy of HIV viral levels after the viral setpoint
 \cite{Perelson1996,Wei1995}. Moreover these models assume a well-mixed environment with no real spatial behaviour. This is very different to what happens at the very earliest stages of infection at the vaginal or rectal epithelium during sexual transmission of HIV from an infected man to his partner. A single lineage usually expands in the new host, even though the genetically heterogeneous inoculum contains numerous infectious units. The dynamics of high HIV seminal loads leading to sporadic infection and the establishment of single foci of infection \cite{Keele2008}, are difficult to understand biologically and completely fall outside the sphere of usual mathematical modelling of infectious diseases with simple ordinary differential equations.



Spatially heterogeneous outcomes may arise from underlying spatial heterogeneity possibly due to tissue architecture or damage from other sexually transmitted infections \cite{Haase2010}, but it is also possible to have such non-uniformities spontaneously arise from the infection dynamics. One of the archetypical manners in which this occurs is through the existence of a Turing instability in a set of partial differential equations (PDEs). Turing's seminal work in 1952 \cite{Turing1952} showed that for some nonlinear reaction-diffusion equations the steady state solution of the system is not spatially uniform. The Turing instability occurs when a spatially homogeneous steady state of the reaction dynamics, which is linearly stable in the absence of diffusion, becomes linearly unstable when the reactions are coupled with the diffusion. This can occur when there are two or more nonlinearly interacting species with different diffusivities. The resultant inhomogeneous spatial pattern is called a Turing pattern.  Turing patterns have been proposed to explain patterning  in numerous physical, chemical and biological systems
 \cite{CM2012},  including models of morphogenesis \cite{MurrayBook,Keshet2005,MWBGL2012,Meinhardt2012} and some chemical reactions \cite{OS1991}. Turing pattern formation has also been investigated in an SIR model to predict the spatial
 transmission of diseases in a population \cite{LJ2007}.

In order to investigate the impact of spatial dynamics in a simple mathematical model of HIV infection
we  extended an SIR model for in-host virus dynamics \cite{Nowak1996,Nowak2000} to include spatially random diffusion and spatially directed chemotaxis \cite{Jin2008}. The spatio-temporal behaviour of this system is investigated within the framework of Turing instability-induced pattern formation \cite{Turing1952} and is shown to result in spatial hot spots of infection for certain ranges of the parameter values. We  further explored the behaviour of the model system through numerical simulations which reveals  complicated spatial dynamics persisting through transient and long time behaviours. These spatial variations have implications for the development and effectiveness of microbicides against HIV \cite{Klasse2006}.

\section{Model Equations}
The standard SIR-based model for in-host virus dynamics is given by \cite{Wei1995,Nowak2000,Perelson1996}
\begin{equation}
  \label{eq:sir}
 \begin{aligned}
\frac{d T}{d t}=&s-kVT-\mu T\\
\frac{d I}{d t}=&kVT-\delta I\\
\frac{d V}{d t}=&N I -c V.
\end{aligned}
\end{equation}
In these equations the dependent variables are:  $T$, the population of uninfected target cells;
$I$, the population of infected cells; and  $V$, the population of free virions, where these all vary with time $t$, but not space.
It is assumed that the target cells are supplied at a constant rate $s$ and they are removed
either through cell death with death rate $\mu$ or by becoming
infected by virions. The parameter $k$ represents
 the  rate of infection of target cells per virion; $N$ represents the
number of virions produced per unit time, per infected cell; $\delta$ is the death rate of the infected cells
and $c$ is the clearance rate of virions.

Extending this model to  also include spatial aspects so that the independent variables are now $(t,\textbf{x})$ produces the following model
\begin{equation}\label{eq:model}
\begin{aligned}
  \frac{\partial T}{\partial t}=&s-kVT-\mu T+D_{T}\nabla^2 T-\chi \nabla(T \nabla I)\\
\frac{\partial I}{\partial t}=&kVT-\delta I+D_{I}\nabla^2 I\\
\frac{\partial V}{\partial t}=&N I-c V+D_{V}\nabla^2 V,
\end{aligned}
\end{equation}
where $T$, $I$ and $V$ are now \emph{concentrations} of target cells, infected cells and virus respectively, with appropriate units according to the space dimension (e.g. $\mathrm{cells/mm}$ for 1D or $\mathrm{cells/mm^{2}}$ for 2D). In this model it is assumed that the target CD4+ T cells,  infected cells and  free virions all diffuse with diffusion constants $D_T$, $D_I$ and $D_V$ respectively. We have also included
a spatial chemotaxis term
$-\chi \nabla(T \nabla I)$ to represent the  chemotactic attraction of  target immune  cells driven by the concentration gradient of cytokines from inflammation at sites of infection. The random walk diffusive motion of these T cells has been well established both {\em in vitro} and {\em in vivo} \cite{Miller2003} with some evidence for an anomalous component to the diffusion \cite{Harris2012}, although in this work, for simplicity, we will assume a purely Brownian diffusion. The chemotaxis of T cells  is more difficult to establish but careful experiments have clearly demonstrated the chemotaxis of T cells in response to gradients of chemokines in microfluidic {\em in vitro} studies \cite{Lin2006}. The vaginal or rectal epithelium will be represented by a 2-dimensional surface so that our space component is given by $\textbf{x}=(x_1,x_2)^T$.

\section{Turing patterns}
We are interested in whether the dynamics of the spatially extended model allow for formation of spatial patterns. The system of partial differential equations \eqref{eq:model} may be classified as a reaction-chemotaxis-diffusion system. Patterns may occur in this system in the neighbourhood of a spatially homogeneous steady state provided the conditions for a Turing instability are met, namely that this spatially homogeneous steady state is:
\begin{itemize}
  \item[(T1)] linearly stable in the absence of diffusion and chemotaxis; and
  \item[(T2)] linearly unstable in the presence of diffusion and chemotaxis.
\end{itemize}

\section{Steady States}
The spatially extended system \eqref{eq:model}, has the same (spatially) homogeneous steady states as
the standard model \eqref{eq:sir}:
 the disease-free state
\begin{align*}
T_0^{*}&=\frac{s}{\mu}, & I_0^{*}&=0, & V_0^{*}&=0,
\end{align*}
and the endemic state
\begin{align*}
T^{*}&=\frac{c \delta}{k N},&I^{*}&=\frac{s k N- c \delta \mu}{N k \delta},&V^{*}&=\frac{s k N- c \delta \mu}{c k \delta}.
\end{align*}

\section{Non-dimensional equations}
In order to reduce the number of parameters and simplify some of the analysis it is useful to work with a non-dimensional version of the system \eqref{eq:model}. Let the non-dimensional dependent variables be given by
\begin{align*}
  u_1 &= T / T_{c}\\
  u_2 &= I / I_{c}\\
  u_3 &= V / V_{c}
\end{align*}
and the non-dimensional independent space and time variables be given by
\begin{align*}
  X_1 &= x_1 / L\\
  X_2 &= x_2 / L\\
  \tau &= t / t_c
\end{align*}
where the values of $T_{c}, I_{c}, V_{c}, L$ and $t_c$ shall be chosen later in such a way to minimise the total number of parameters.
We substitute the non-dimensional variables into \eqref{eq:model} and obtain

\begin{equation}
\begin{aligned}
  \frac{\partial{u_1}}{\partial \tau} &= \left( \frac{st_c}{T_{c}} \right) - \left( V_{c}t_c k \right) {u_3}{u_1}  - \left( \mu t_c \right){u_1} + \left( \frac{D_T t_c}{L^2} \right) \nabla^2{u_1} - \left( \frac{t_c I_{c}\chi}{L^2} \right) \nabla \left( {u_1}\nabla {u_2} \right)\\
    \frac{\partial{u_2}}{\partial \tau} &= \left( \frac{kV_{c}T_{c}t_c}{I_{c}} \right)u_1 u_3 - \left( \delta t_c \right){u_2} + \left( \frac{D_I t_c}{L^2} \right)\nabla^2{u_2}\\
  \frac{\partial{u_3}}{\partial \tau} &= \left( \frac{NI_{c} t_c}{V_{c}}\right){u_2} - \left( c t_c \right){u_3} + \left( \frac{D_V t_c}{L^2} \right) \nabla^2 {u_3}.
\end{aligned}
\end{equation}

There is now freedom in setting the scaling factors and simplifying the equations. To this end set $T_{c}$, $I_{c}$ and $V_{c}$ to the corresponding values of the endemic homogeneous steady state $T^*, I^*, V^*$.
The non-dimensional equations then become
\begin{equation}\label{eq:nondim1}
\begin{aligned}
  \frac{\partial{u_1}}{\partial \tau} &= t_c \left( \frac{skN}{c\delta} \right)(1- {u_3}{u_1}) - t_c\mu {u_1}(1 - {u_3}) + \left( \frac{D_T t_c}{L^2} \right) \nabla^2{u_1}\\ &\quad  - t_c \left( \frac{ (skN - c\delta\mu)\chi}{Nk\delta L^2} \right) \nabla \left( {u_1}\nabla {u_2} \right)\\
  \frac{\partial{u_2}}{\partial \tau} &= (\delta t_c){u_1}{u_3} - \left( \delta t_c \right){u_2} + \left( \frac{D_I t_c}{L^2} \right)\nabla^2{u_1}\\
  \frac{\partial{u_3}}{\partial \tau} &= (c t_c){u_2} - \left( c t_c \right){u_3} + \left( \frac{D_V t_c}{L^2} \right) \nabla^2 {u_3}.
\end{aligned}
\end{equation}

We may further simplify \eqref{eq:nondim1} by choosing $t_c := 1/\mu$ and $L^2 := D_T/\mu$. Then we set the new non-dimensional parameters to be the following
\begin{align*}
  \xi &:= \frac{skN}{c\delta\mu}, &  d_{I} &:= \frac{D_I}{D_T}, &  d_{V} &:= \frac{D_V}{D_T},\\
  d_{\chi} &:= \frac{s\chi}{\delta D_T}\left( 1 - \frac{1}{\xi} \right), &  \alpha &:= \frac{\delta}{\mu}, &  \beta &:= \frac{c}{\mu}.
\end{align*}
Finally, we arrive at the  non-dimensional version of \eqref{eq:model}:
\begin{equation}\label{eq:model_nde}
  \begin{aligned}
  \frac{\partial{u_1}}{\partial \tau} &= \xi - (\xi-1)u_1u_3 - u_1 + \nabla^2{u_1} - d_{\chi}\nabla \left( {u_1}\nabla {u_2} \right)\\
  \frac{\partial{u_2}}{\partial \tau} &=\alpha({u_1}{u_3} - {u_2}) + d_I \nabla^2{u_2}\\
  \frac{\partial{u_3}}{\partial \tau} &= \beta({u_2} -{u_3}) + d_{V} \nabla^2 {u_3}.
\end{aligned}
\end{equation}
It is straightforward to check that in the absence of spatial variation this system admits the following two equilibria:
\begin{itemize}
  \item[(E1)] The endemic spatially homogeneous steady-state at  $(u_1^*, u_2^*, u_3^*) = (1,1,1)$; and
  \item[(E2)] The disease-free spatially homogeneous steady state at $(u_1^*, u_2^*, u_3^*) = (\xi, 0, 0)$.
\end{itemize}

Observe also that the value of $\xi$ in \eqref{eq:model_nde} determines one of two distinct scenarios:
\begin{itemize}
  \item[(i)] When $\xi < 1$ we have $I,V < 0$ whenever $u_2, u_3 \ge 0$, hence the endemic homogeneous steady state is not physically relevant as it corresponds to negative concentrations of infected cells and free virus. The only spatially homogeneous steady state is the state free of disease, $(\xi, 0, 0)$.
  \item[(ii)] When $\xi > 1$ both of the spatially homogeneous steady states, $(\xi, 0, 0)$ and $(1,1,1)$, are physically relevant.
\end{itemize}

In order to explore the possibility of Turing pattern formation we shall consider the stability of \eqref{eq:model_nde} linearised about each of the two steady states.

\section{Linearising about the endemic spatially homogeneous steady state}

Let us assume for this section that $\xi>1$ as we have already established above that this is a necessary condition for a physically relevant endemic homogeneous steady state. We perturb this equilibrium $(u_1^{*}, u_2^*, u_3^*) = (1,1,1)$ by writing $u_i = u_i^{*}+\Delta u_i$ for each $i\in\{1,2,3\}$.
The linearised form of \eqref{eq:model_nde} for the perturbations $\Delta {u_i}$ is given by
\begin{equation}
  \frac{\partial}{\partial \tau}
  \begin{bmatrix}
    \Delta{u_1}\\ \Delta{u_2}\\ \Delta{u_3}
  \end{bmatrix}  =
  \begin{bmatrix}
    -\xi & 0 & -(\xi-1)\\
    \alpha & -\alpha & \alpha\\
    0 & \beta & -\beta
  \end{bmatrix}   \begin{bmatrix}
    \Delta{u_1}\\ \Delta{u_2}\\ \Delta{u_3}
  \end{bmatrix} +
  \begin{bmatrix}
    1 & -d_{\chi} & 0\\
    0 & d_I & 0\\
    0 & 0 & d_V
  \end{bmatrix} \nabla^2 \begin{bmatrix}
    \Delta{u_1}\\ \Delta{u_2}\\ \Delta{u_3}
  \end{bmatrix}.
  \label{eq:nondimlin}
\end{equation}
It shall be convenient to carry out Fourier transforms with respect to the  spatial variables in equation \eqref{eq:nondimlin}. This yields
\begin{equation}
   \frac{\partial}{\partial \tau}
  \begin{bmatrix}
    \Delta{U_1}\\ \Delta{U_2}\\ \Delta{U_3}
  \end{bmatrix}  =
  \begin{bmatrix}
    (-\xi - q^2) & d_{\chi} q^2 & -(\xi-1)\\
    \alpha & (-\alpha - d_Iq^2) & \alpha\\
    0 & \beta & (-\beta - d_Vq^2)
  \end{bmatrix}   \begin{bmatrix}
    \Delta{U_1}\\ \Delta{U_2}\\ \Delta{U_3}
  \end{bmatrix},
  \label{eq:fourierlin}
\end{equation}
where $\Delta U_i$ denote the Fourier transform of $\Delta u_i$ and $q^2 := \mathbf{q}^T \mathbf{q}$ where $\mathbf{q}\in\mathbb{R}^{d}$ is the Fourier variable, with $d$ the spatial dimension. The conditions for Turing instabilities are determined from the eigenvalue spectrum of the coefficient matrix in \eqref{eq:fourierlin}. The requirement that the homogeneous steady state is stable in the absence of diffusion and chemotaxis is met if all eigenvalues have negative real parts when $q^2=0$. A Turing instability may then occur if one or more eigenvalues have positive real parts for some $q^2>0$.

The characteristic polynomial of the matrix in \eqref{eq:fourierlin} is
\begin{equation}\label{eq:chr}
\begin{aligned}
  &\lambda^3 + \lambda^2(\xi + \alpha + \beta + (1+d_I + d_V)q^2)) \\
  + & \lambda\left\{\xi\alpha + \xi\beta + (\alpha d_V + \beta d_I + \xi d_I + \alpha + \xi d_V + \beta - \alpha d_\chi)q^2\right.\\
  + & \left.(d_I d_V + d_I + d_V )q^4 \right\}\\
  + & \left\{ (\xi\alpha d_V + \xi\beta d_I - \alpha\beta d_\chi)q^2 + (\xi d_I d_V + \alpha d_V + \beta d_I - \alpha d_\chi d_V)q^4\right.\\
  + & \left.(d_I d_V)q^6 + (\xi-1)\alpha\beta \right\}
\end{aligned}
\end{equation}

To check for stability, recall that a cubic polynomial $\lambda^3 + a\lambda^2 + b\lambda + c$ has all roots in the left half complex planeif and only if $a,b,c > 0$ and $ab > c$ (Ruth--Hurwitz stability criterion for a cubic polynomial \cite{Hurwitz1964}). Below we consider the characteristic polynomial for different values of $q^2$.

\subsection{Case of no spatial variation}

In the absence of diffusion and chemotaxis ($q^2 = 0$) the characteristic polynomial \eqref{eq:chr} simplifies to
\begin{equation}
  \lambda^3 + \lambda^2(\alpha + \beta + \xi) + \lambda(\alpha\xi + \beta\xi) +(\xi-1)\alpha\beta.\label{eq:chp_homogeneous}
\end{equation}

Since we have assumed that $\xi>1$, all the coefficients of \eqref{eq:chp_homogeneous} are positive. Thus the remaining condition
\begin{equation}
  (\alpha+\beta+\xi)(\alpha+\beta)\xi > (\xi-1)\alpha\beta
  \label{eq:hom_st_cond}
\end{equation}
is necessary and sufficient for all roots of \eqref{eq:chp_homogeneous} to be in the left half-plane.

Since $\alpha$ and $\beta$ are positive and $\xi>1$, equation \eqref{eq:hom_st_cond} is always true:
\begin{align*}
  (\alpha + \beta + \xi)(\alpha + \beta)\xi &> (\alpha+\beta)^2\xi\\
  & > 2\alpha\beta\xi\\
  & > (\xi-1)\alpha\beta.
\end{align*}
Thus the endemic steady state is stable in the absence of diffusion and chemotaxis and the first Turing condition (T1) is satisfied.

\subsection{Case of spatial variation}
Now we investigate the nature of the roots of \eqref{eq:model_nde} in the presence of diffusion and chemotaxis ($q^2 > 0$).

Firstly, consider the case for large $q^2$. Then the characteristic polynomial \eqref{eq:chr} will have all its coefficients positive (as in each term the largest power of $q^2$ has a positive coefficient). Moreover, the coefficient of $\lambda^2$ tends to $(1+d_I+d_V)q^2$, the coefficient of $\lambda$ tends towards $(d_{I}d_{V} + d_{I} + d_{V})q^4$ while the constant coefficient is dominated by $(d_{I}d_{V})q^{6}$. It is then straightforward to show that the product of the former two is larger than the latter, hence for sufficiently large $q^2$ all roots of \eqref{eq:chr} are in the left half complex plane and the system is stable.

Even though for large values of $q^2$ the system becomes stable, there may still be sufficiently large values of $d_{\chi}$ and an appropriate range of values of $q^2$ for which instabilities occur.
To determine the threshold on $d_{\chi}$ (or $\chi$) above which instabilities may be possible, let us first write the polynomial \eqref{eq:chr} in the following form:
\begin{equation}
  \label{eq:chp_simple}
  \lambda^3 + \lambda^2(a_1 + a_2 q^2) + \lambda(b_1 + b_2q^2 + b_3q^4) + (c_1+c_2q^2 + c_3q^4 + c_4q^6)
\end{equation}
for appropriate values of coefficients $a_{1}, a_{2}, b_{1}, b_2, b_3, c_1, c_2, c_3$ and $c_4$. Note that all these coefficients are positive except for possibly $b_2$, $c_2$ and $c_3$. If we assume that these too are positive then:
\begin{itemize}
  \item[(S1)] $a_1b_1 > c_1$\\
  \item [(S2)] $a_1b_2 + a_2b_1 > c_2$\\
  \item [(S3)] $a_1b_3 + a_2b_2 > c_3$\\
  \item [(S4)] $a_2b_3 > c_4$.
\end{itemize}
The proof of this is straightforward but tedious, and is shown in the Appendix.
Now, (S1)--(S4) imply that
\begin{displaymath}
  (a_1 + a_2q^2)(b_1 + b_2q^2 + b_3q^4) > c_1 + c_2q^2 + c_3q^4 + c_4q^6,
\end{displaymath}
for all $q$, therefore by the Ruth-Hurwitz stability criterion, all roots of \eqref{eq:chp_simple} (and \eqref{eq:chr}) are in the left half complex plane and thus \eqref{eq:fourierlin} is stable.

We conclude that the only way the system in \eqref{eq:fourierlin} may become unstable is if at least one of $b_2$, $c_2$ or $c_3$ becomes non-positive.
Therefore the necessary (but not sufficient) condition for Turing condition (T2) to hold is that either
\begin{itemize}
\item[(C1)] $b_2 \le 0$, that is $d_\chi \ge d_V + (\beta + \xi) d_I / \alpha + 1 + \xi + 1/\alpha$;
\item[(C2)] $c_2 \le 0$, that is $d_\chi \ge \xi d_V / \beta + \xi d_I/\alpha$; or
\item[(C3)] $c_3 \le 0$, that is $d_\chi \ge \xi d_I / \alpha + 1 + \beta d_I / (\alpha d_V) $.
\end{itemize}
Using $d_{\chi}  = \frac{s\chi}{\delta D_T}\left( 1 - \frac{1}{\xi} \right)$ we see that (C1), (C2) and (C3) are equivalent to
\begin{itemize}
  \item[(C1')] $\chi \ge \frac{\delta D_T}{s\chi \left( 1-1/\xi \right)}\left( d_V + (\beta + \xi) d_I / \alpha + 1 + \xi + 1/\alpha \right)$,
  \item[(C2')] $\chi \ge \frac{\delta D_T}{s\chi \left( 1-1/\xi \right)}\left(\xi d_V / \beta + \xi d_I/\alpha\right)$, or
  \item[(C3')] $\chi \ge \frac{\delta D_T}{s\chi \left( 1-1/\xi \right)} \xi d_I / \alpha + 1 + \beta d_I / (\alpha d_V)$
\end{itemize}
respectively. Given a set of parameters, for the purpose of determining a possible existence of a Turing instability we shall only consider the weakest condition of the three.

\section{Linearising about the disease-free spatially homogeneous steady state}
Similar to our approach for the endemic equilibrium in the previous section, we may linearise \eqref{eq:model_nde} about the disease-free spatially homogeneous steady state $(u_1^{*}, u_2^*, u_3^*) = (\xi,0,0)$. Since this state is always physically relevant, for the time being we need not assume any additional conditions on $\xi$ (apart from positivity). After performing a Fourier transform of \eqref{eq:model_nde} linearized about $(\xi, 0,0)$,  we get
\begin{equation}\label{eq:no_disease}
  \frac{\partial}{\partial \tau}
  \begin{bmatrix}
    \Delta{U_1}\\ \Delta{U_2}\\ \Delta{U_3}
  \end{bmatrix}  =
  \begin{bmatrix}
    -1-q^2 & d_{\chi}\xi q^2 & -\xi(\xi-1)\\
    0 & (-\alpha-d_Iq^2) & \alpha\xi\\
    0 & \beta & (-\beta-d_Vq^2)
  \end{bmatrix}
  \begin{bmatrix}
    \Delta{U_1}\\ \Delta{U_2}\\ \Delta{U_3}
  \end{bmatrix}.
\end{equation}

In the spatially homogeneous setting ($q^2=0$), it is easy to check that the disease-free steady state is stable if and only if $\xi < 1$.

The matrix in \eqref{eq:no_disease} has $-1-q^2 < 0$ as one of its eigenvalues. The remaining two are given by the eigenvalues of the $2\times 2$ lower-right submatrix
\begin{displaymath}
  \begin{bmatrix}
    (-\alpha-d_Iq^2) & \alpha\xi\\
    \beta & (-\beta-d_Vq^2).
  \end{bmatrix}
\end{displaymath}
It is easy to check that if $\xi<1$ the eigenvalues of this matrix lie in the left half complex plane (e.g. observe that the matrix has negative trace and positive determinant), and this holds for all values of $q^2$. So (T1) and  (T2) cannot both be true, and we conclude that the Turing conditions cannot be satisfied in a neighbourhood of the disease-free steady state.

\section{Regularisation of chemotaxis}
It is well-known that in two spatial dimensions chemotaxis above a certain threshold results in a finite time blow up of solutions to the governing equations (see e.g. \cite{Horstmann2003}). The standard approach to deal with this non-realistic phenomenon  is to introduce a regularisation term to \eqref{eq:model_nde}.  We shall use a density-dependent sensitivity regularisation, studied in \cite{Velazquez2004}, based on the assumption that with increasing cell density, their advective velocity reduces. For other forms of regularisation see the survey article by Hillen and Painter \cite{Hillen2009}. Our new governing equations become:

\begin{equation}\label{eq:model_reg}
  \begin{aligned}
    \frac{\partial{u_1}}{\partial \tau} &= \xi - (\xi-1)u_1u_3 - u_1 + \nabla^2{u_1} - (1+\epsilon)d_{\chi}\nabla \left( \frac{u_1}{1+\epsilon u_1}\nabla {u_2} \right)\\
  \frac{\partial{u_2}}{\partial \tau} &=\alpha({u_3}{u_1} - {u_2}) + d_I \nabla^2{u_2}\\
  \frac{\partial{u_3}}{\partial \tau} &= \beta({u_2} -{u_3}) + d_{V} \nabla^2 {u_3}
\end{aligned}
\end{equation}
for some dimensionless regularisation parameter $\epsilon \ge 0$ such that in the limit as $\epsilon\to 0$ we recover the original non-regularised model.

Let us introduce the notion of effective chemotaxis:
\begin{displaymath}
  \tilde{d}_\chi(u_1; \epsilon) := \frac{1+\epsilon}{1+\epsilon u_1} d_\chi.
\end{displaymath}
Then $\tilde{d}_{\chi} = d_{\chi}$ at the endemic steady state ($u_1=1$), and the linearised equations \eqref{eq:nondimlin} and \eqref{eq:fourierlin} remain unchanged, therefore in this case our stability analysis and conditions for Turing patterns formation from previous sections also apply in the regularised setting.
Also note that $\tilde{d}_{\chi}\to0$ as $u_1\to\infty$.

\section{Parameters, units and dimensions}

There has been considerable investigation of parameter values for variants of the SIR model, based on trials in HIV-infected individuals  \cite{Perelson1996,Perelson1997,Murray2007,Murray2011} and from SIV infected macaques \cite{Mandl2007}, and it is reasonable to assume that these parameter values  provide useful starting approximations for variants of the model \eqref{eq:sir}, that also includes a spatial component.

We assume the original HIV model \eqref{eq:sir} has parameter estimates given by
$N=480 \;\mathrm{virions\; cell^{-1}\;day^{-1}}$, $k=3.43\times10^{-5}\;\mathrm{ml\;virions^{-1}\;day^{-1}}$, $\delta=0.5\;\mathrm{day^{-1}}$, $c=3\;\mathrm{day^{-1}}$, $s=10\;\mathrm{cells\;mm^{-3}\;day^{-1}}$, $\mu=0.03\;\mathrm{day}^{-1}$.

We will perform numerical simulations in both one and two spatial dimensions. Note that the constants $k$ and $s$ given above are volume-based. In order to adapt them to two or one spatial dimensions let us assume that the region in which we are solving the PDE is either a thin rectangular sheet of thickness $h$ or a thin wire with a square $h\times h$ cross-section.
Thus the new dimension-specific values of $k$ and $s$ become $\tilde{k} = kh^{3-d}$ and $\tilde{s} = sh^{3-d}$ respectively, where $d\in\{1,2\}$ is the number of spatial dimensions. Note that by changing $d$ the only non-dimensional parameter that changes is $d_{\chi}$. For all of the numerical simulations we shall take $h=0.1\;\mathrm{mm}$.

The diffusion of T cells in lymphatic tissue has been estimated at
\begin{equation}
  D_{T}=1.1\; \mathrm{\mu m^{2}\;s^{-1}} = 0.09504\;\mathrm{mm^2\;day^{-1}}
\end{equation}
It can be assumed that the uninfected ($T$) and infected ($I$) CD4+ T cells will have similar diffusion coefficients. Different studies have shown virions with diffusion coefficients of the order of $0.0088 \mathrm{\mu m^{2}\;s^{-1}}$.
The only parameters without experimental bounds are the effective chemotaxis term $\chi$ and the regularisation constant $\epsilon$. For the above stated parameters, the weakest necessary (but not sufficient) condition for Turing instability is $(C2)$, requiring $d_{\chi} > 219.8$, that is $\chi > 10.4\;\mathrm{mm^4\;cell^{-1}\;day^{-1}}$ in two spatial dimensions and $\chi > 104\;\mathrm{mm^3\;cell^{-1}\;day^{-1}}$ in one spatial dimension. These are our lower bounds on the chemotactic threshold.

Figure \ref{lambdavsq} shows the real part of the leading eigenvalue of the matrix of the linearized system \eqref{eq:fourierlin} for the parameters stated above as a function of spatial frequency. Note that even though $\chi = 110 \mathrm{\;mm^3\;cell^{-1}\;day^{-1}}$ is above our calculated lower bound on the chemotactic threshold, the system remains stable for all spatial frequencies. Turing conditions are only met for $\chi$ somewhere between $110$ and $120\mathrm{\;mm^3\;cell^{-1}\;day^{-1}}$.

Recall that we set the scaling factors for independent variables to $t_c = 1/\mu$ and $L^2 = D_T/\mu$. With the given parameters, this results in values $t_c \approx 33.3\; \mathrm{days}$ and $L \approx 1.78\;\mathrm{mm}$.

\begin{figure}
\begin{center}
\includegraphics[width=1\columnwidth]{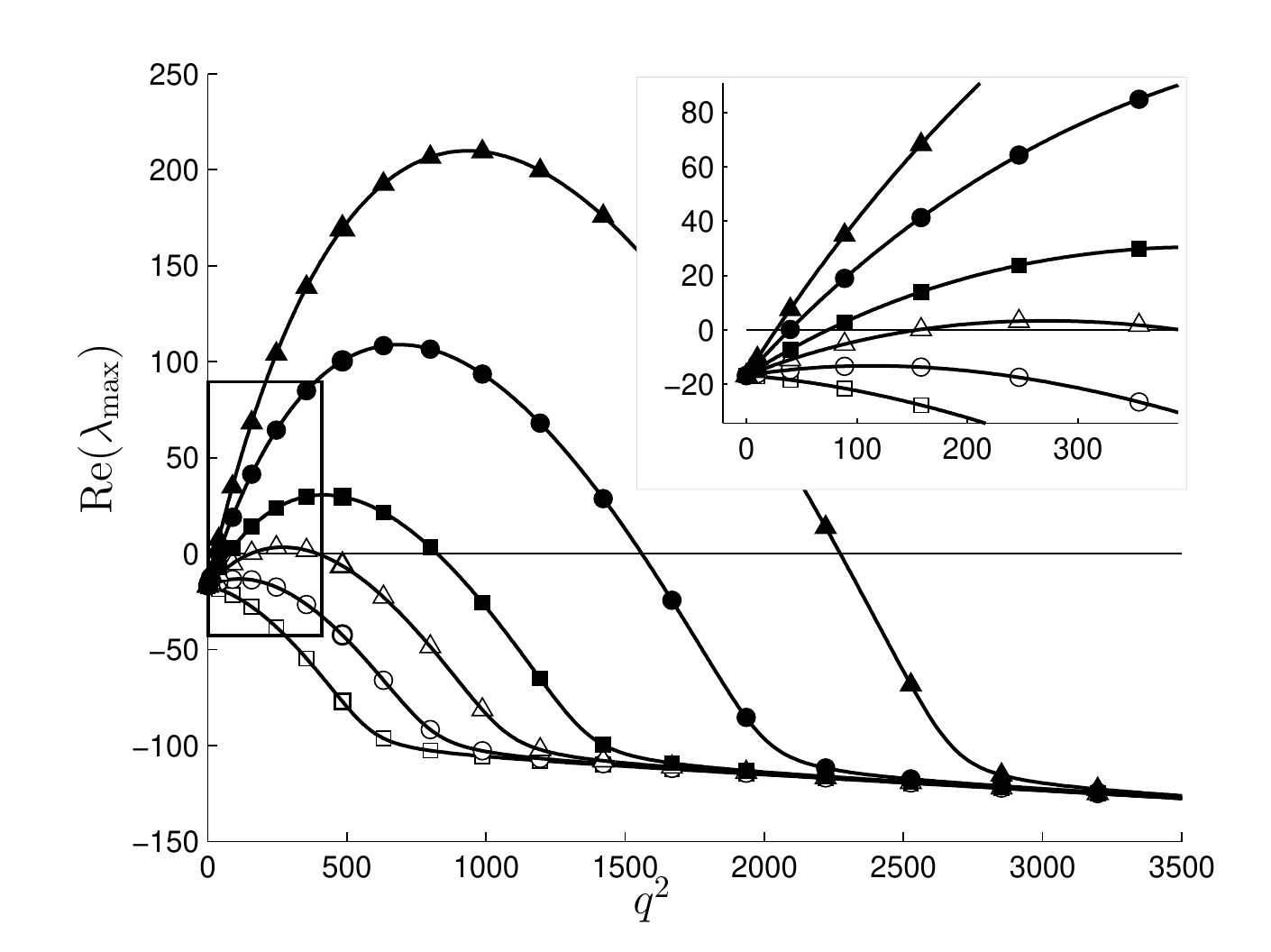}
\caption{Leading eigenvalues of the linear model in one spatial dimension as a function of spatial frequency for several values of $\chi$. Empty squares correspond to $\chi = 100 \mathrm{\;mm^3\; cell^{-1}\; day^{-1}}$; empty circles to $\chi = 110 \mathrm{\;mm^3\; cell^{-1}\; day^{-1}}$; empty triangles to $\chi = 120 \mathrm{\;mm^3\; cell^{-1}\; day^{-1}}$; filled squares to $\chi = 130 \mathrm{\;mm^3\; cell^{-1}\; day^{-1}}$; filled circles to $\chi = 150 \mathrm{\;mm^3\; cell^{-1}\; day^{-1}}$; and filled triangles to $\chi = 170 \mathrm{\;mm^3\; cell^{-1}\; day^{-1}}$. The upper right panel shows an enlarged view around the origin.}
\label{lambdavsq}
\end{center}
\end{figure}

\section{Numerical solutions in one spatial dimension}

The model PDE is solved numerically on a one-dimensional domain of length $L$ with zero-flux boundary conditions. In one-dimension the chemotaxis term need not be regularised, thus, unless otherwise stated, we shall use $\epsilon=0$ in this section. We used MATLAB's built in 1D PDE solver {\tt pdepe}. The initial conditions are taken as either a small random or deterministic perturbation around the endemic steady state.

For $\chi$ below $104\mathrm{\;mm^3\;cell^{-1}\;day^{-1}}$ (our lower bound on the chemotactic threshold for a Turing instability) the solutions become flat with no spatial features as time progresses. This can be seen in Figure \ref{fig:chemotaxis1d_chi100}.

\begin{figure}
\begin{center}
\includegraphics[width=1\columnwidth]{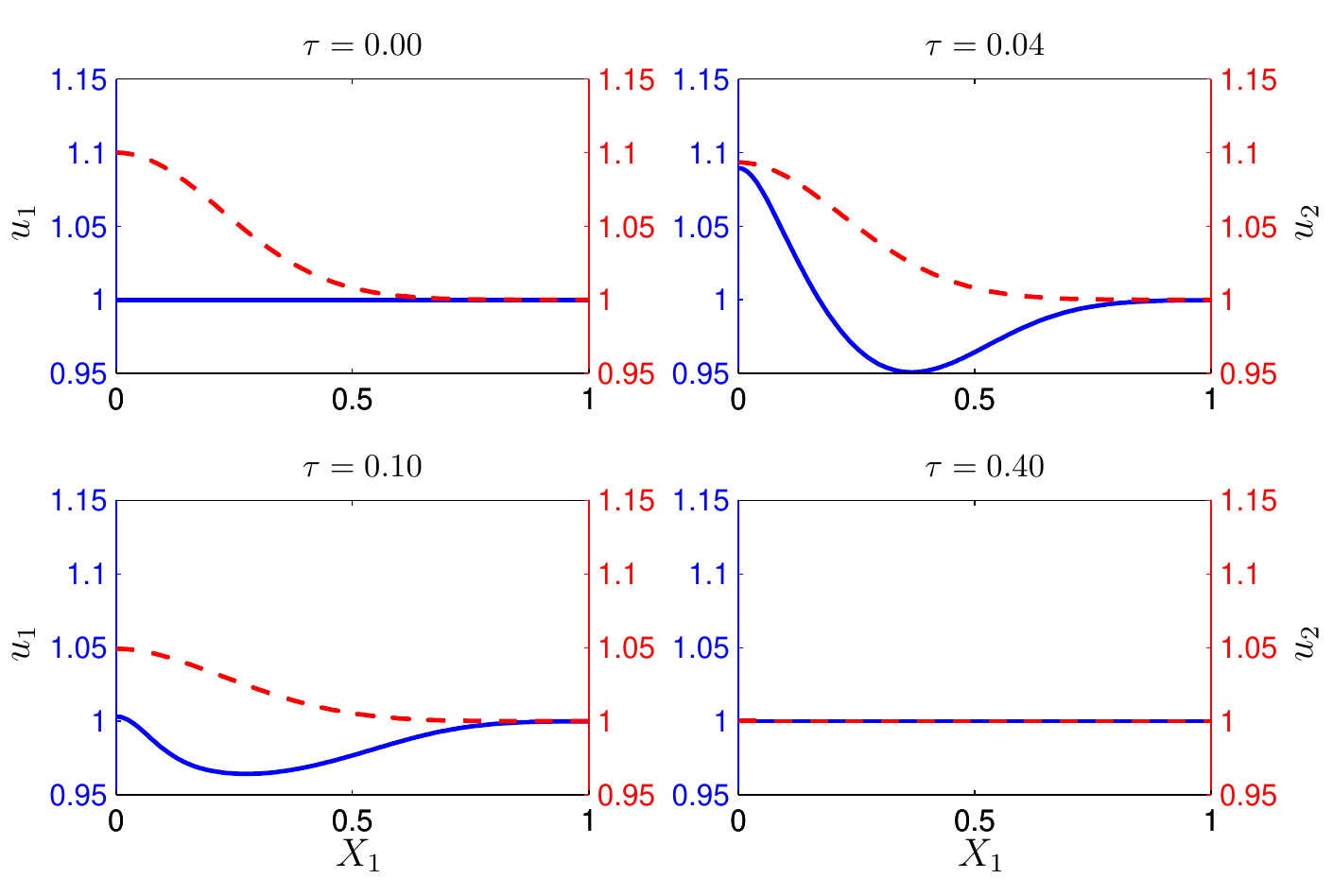}
\caption{Solution for $\chi = 100\mathrm{\;mm^3\;cell^{-1}\;day^{-1}}$. The solid blue line represents $u_1$ (left axis), the non-dimensional version of concentration of healthy cells and the dashed red line is $u_2$ (right axis), the concentration of infected cells. The concentration of infected cells $u_2$ is initially normally distributed. The concentration of free virus (not shown in the graph) is initially constant at $u_3 = 1$ and changes similarly to $u_2$.}
\label{fig:chemotaxis1d_chi100}
\end{center}
\end{figure}

For $\chi$ above this threshold, we see Turing pattern formation in the form of three uniformly separated peaks. This pattern remains and the peaks do not blow up in time; see Figure \ref{fig:chemotaxis1d_chi130}.

\begin{figure}
\begin{center}
\includegraphics[width=1\columnwidth]{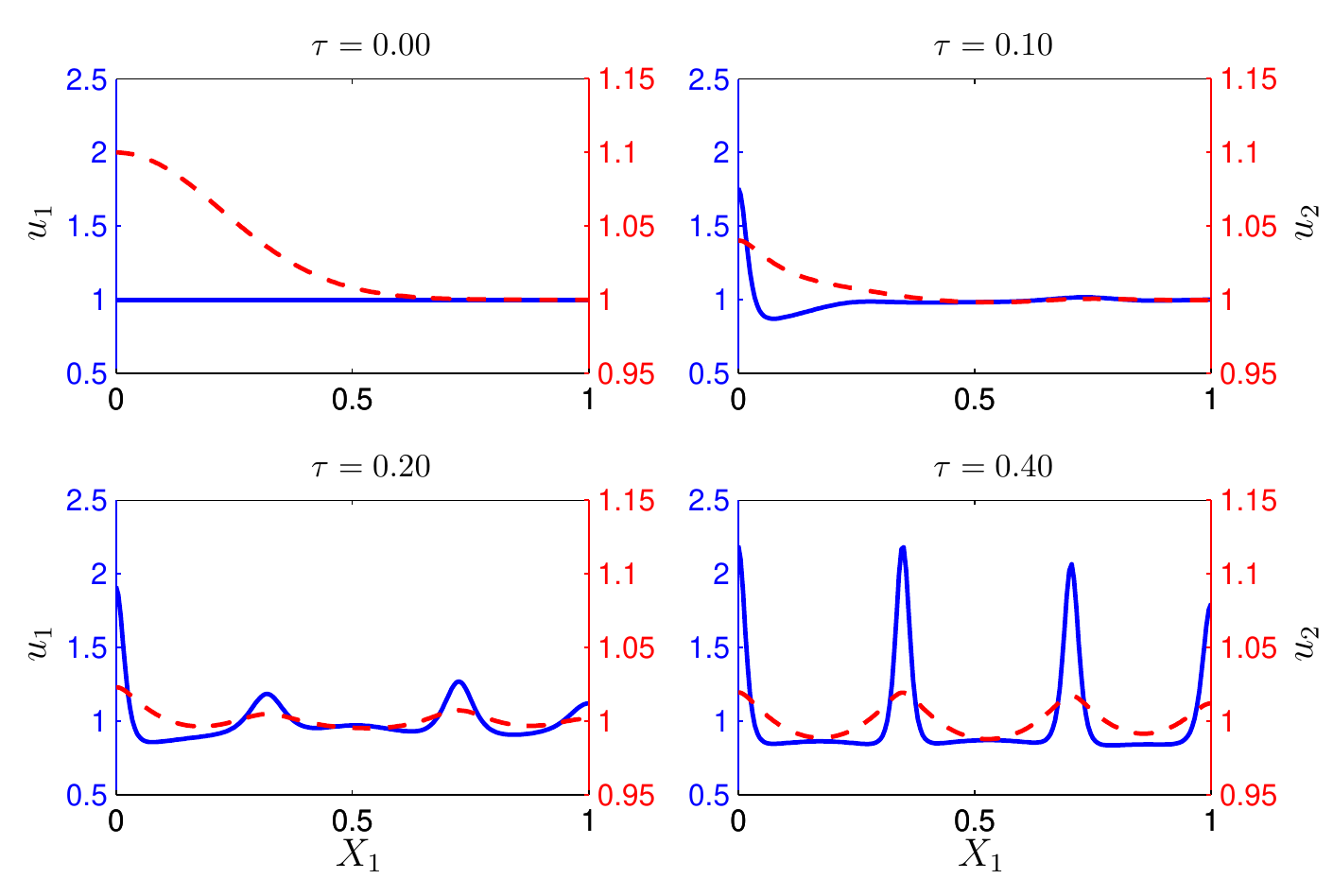}
\caption{Solution for $\chi = 130\mathrm{\;mm^3\;cell^{-1}\;day^{-1}}$ with same initial conditions as those in Figure \ref{fig:chemotaxis1d_chi100}. The solid blue line represents $u_1$ (left axis), the non-dimensional version of concentration of healthy cells and the dashed red line is $u_2$ (right axis), the concentration of infected cells.}
\label{fig:chemotaxis1d_chi130}
\end{center}
\end{figure}

As the strength of chemotaxis increases, we see that the spikes become steeper. However the solutions are still finite for all time; see Figure \ref{fig:chemotaxis1d_chi220_rand}.

\begin{figure}
\begin{center}
\includegraphics[width=1\columnwidth]{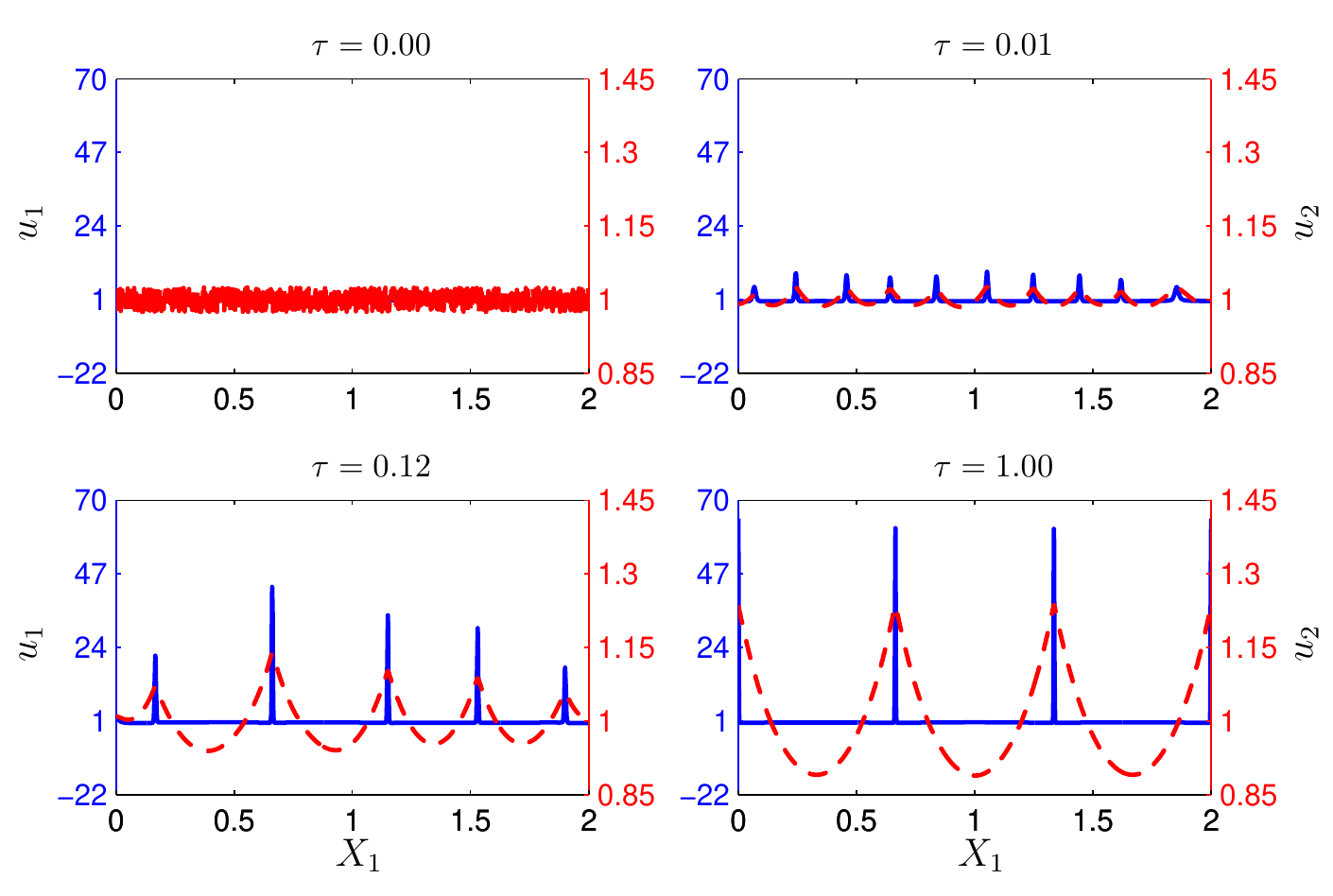}
\caption{Solution for $\chi = 220\mathrm{\;mm^3\;cell^{-1}\;day^{-1}}$ with $u_1 = u_3 = 1$ initially constant and $u_2$ taken from a uniform random distribution in $(0.975, 1.025)$. Note that the size of the domain has been doubled compared to the preceding figures. The solid blue line represents $u_1$ (left axis), the non-dimensional version of concentration of healthy cells and the dashed red line is $u_2$ (right axis), the concentration of infected cells.}
\label{fig:chemotaxis1d_chi220_rand}
\end{center}
\end{figure}

Increasing the length of the domain results in proportionally more spikes; see Figure \ref{fig:chemotaxis1d_chi130_long}. It should also be noted that in this case it takes a longer length of time for the solutions to settle to a steady state and the spikes to become uniform in height.
\begin{figure}
\begin{center}
\includegraphics[width=1\columnwidth]{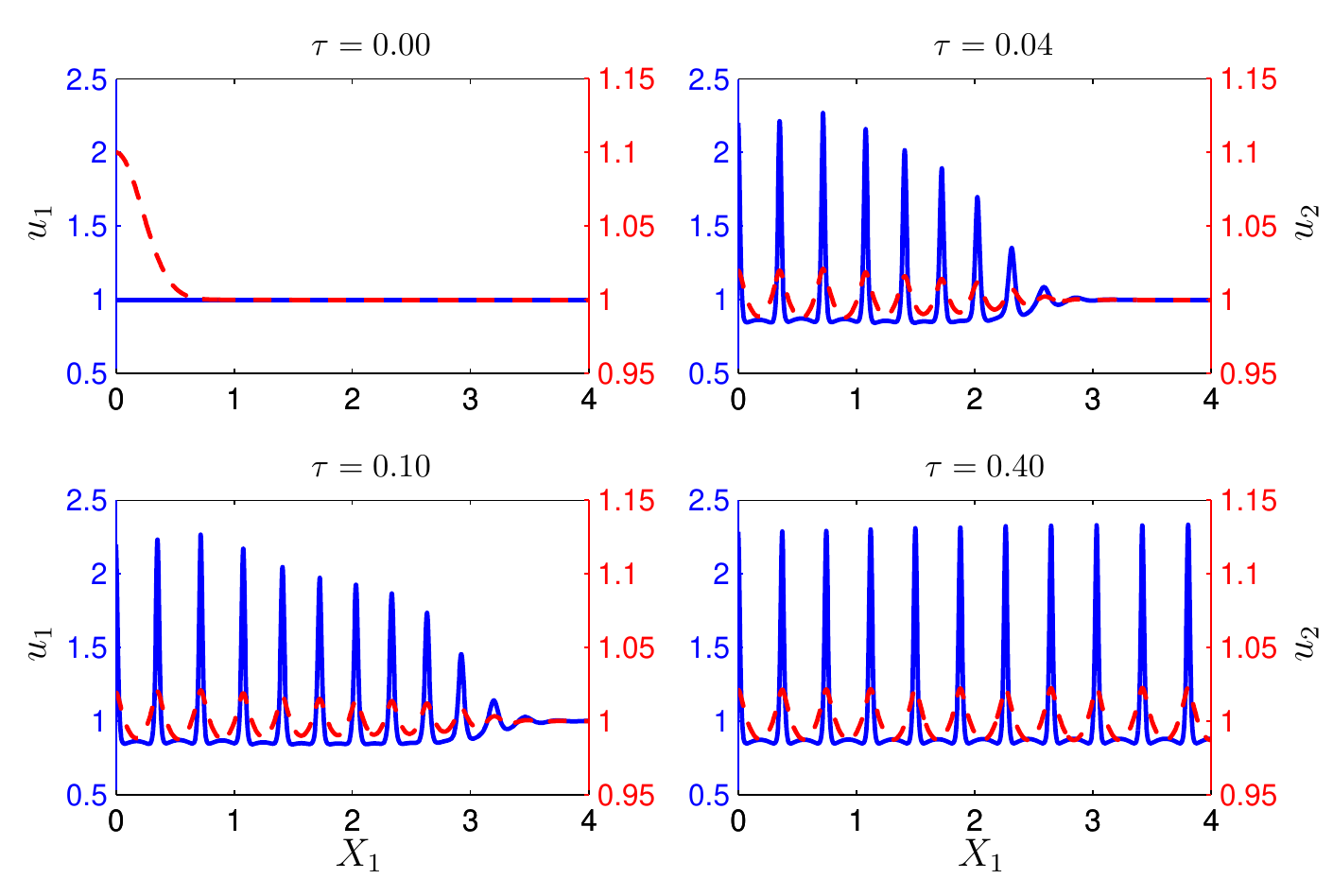}
\caption{Solution for $\chi = 130\mathrm{\;mm^3\;cell^{-1}\;day^{-1}}$ on a longer domain. The speed of propagation of the pattern is estimated to be $0.36\;\mathrm{mm\;day^{-1}}$. The solid blue line represents $u_1$ (left axis), the non-dimensional version of concentration of healthy cells and the dashed red line is $u_2$ (right axis), the concentration of infected cells.}
\label{fig:chemotaxis1d_chi130_long}
\end{center}
\end{figure}

Starting with a random initial condition of $u_2$ (e.g. uniform on $(0.975,1.025)$), the oscillations are initially at a higher frequency than for a smooth initial condition. However these oscillations quickly settle and produce a pattern similar to the case with normal initial distribution; see Figure \ref{fig:chemotaxis1d_chi220_rand}.

%

\subsection{Dependence on initial conditions}

The spatial frequency of the steady state solution may depend on the initial conditions for long domains. Figure \ref{fig:chemotaxis1d_init_dep} shows density profiles for two solutions with different initial distributions of infected cells --- the first one concentrated on the left border and second one concentrated in the middle of the region. The resulting final states contain a different number of peaks --- $16$ and $17$, respectively.

\begin{figure}[t]
  \begin{center}
    \includegraphics[width=1\columnwidth]{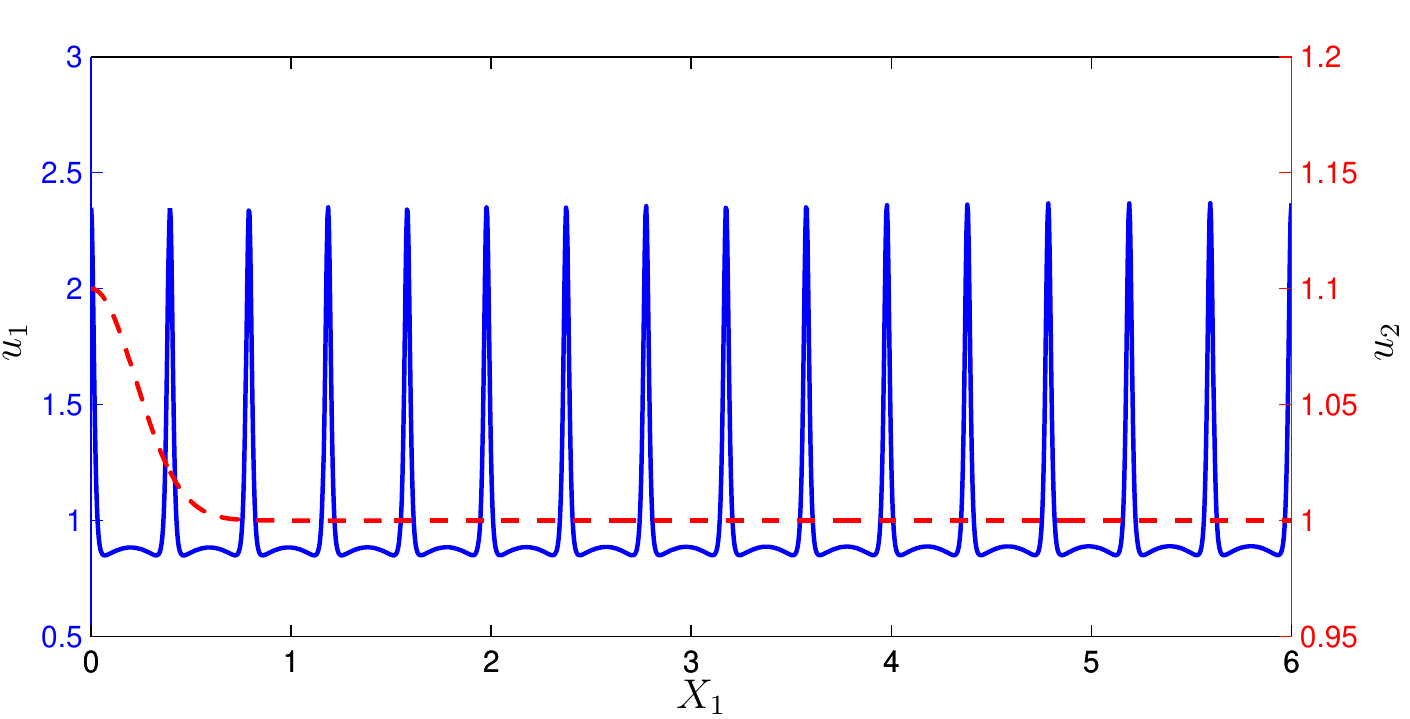}
    \includegraphics[width=1\columnwidth]{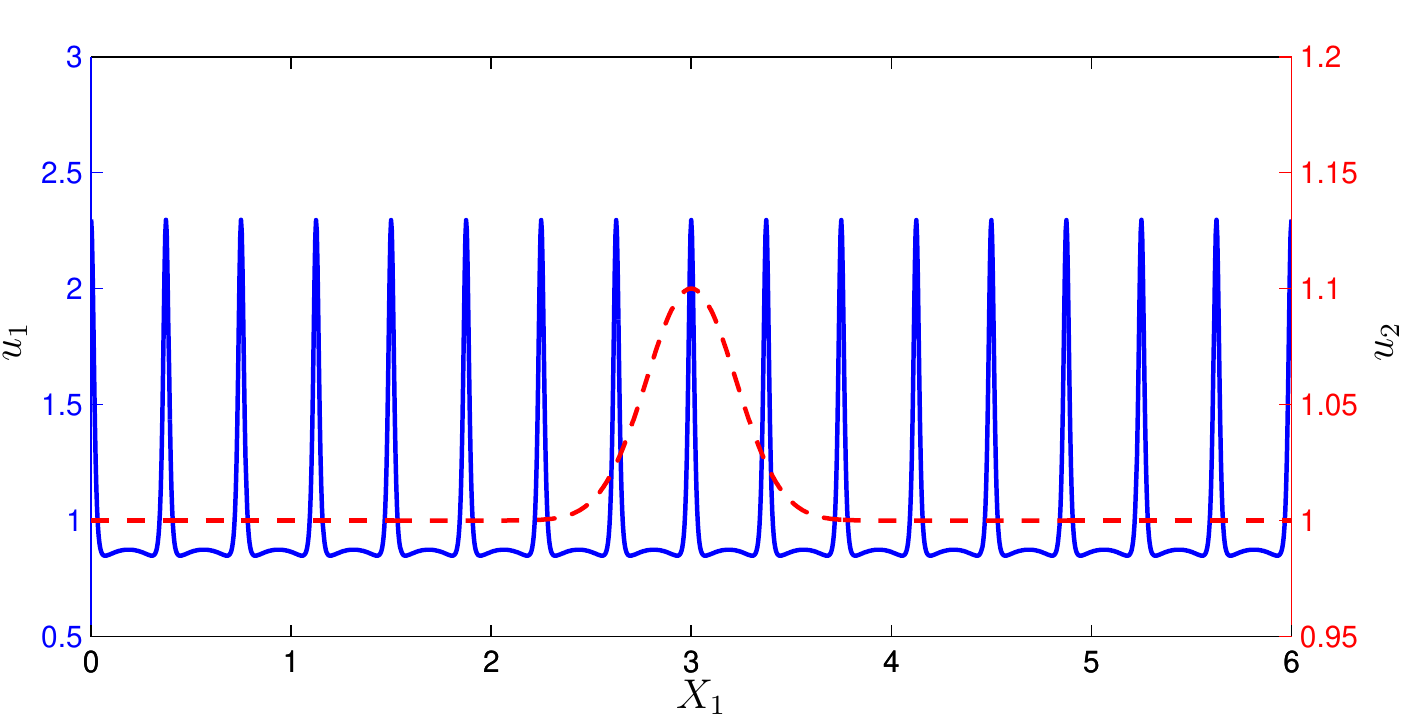}
  \end{center}
  \caption{Solution for $\chi = 130\mathrm{\;mm^3\;cell^{-1}\;day^{-1}}$. Different spatial frequencies resulting from different initial conditions. The red dashed lines indicate the initial values of $u_2$ and the blue unbroken lines indicate the final state of $u_1$. The other initial values are $u_1 = u_3 = 1$. Note the different number of peaks in the final states.}
  \label{fig:chemotaxis1d_init_dep}
\end{figure}

\subsection{Solutions with initial conditions near the disease-free steady state.}
We have already shown that the disease-free homogeneous steady-state is unstable when $skN>c\delta\mu$. Hence we can expect any disturbance of this state (i.e. the introduction of virus or infected cells) to result in system transitioning to the endemic steady state.

Because we are considering global behaviour of the system it is more useful to draw plots  in terms of re-dimensionalised variables (for example, it enables one to compare the number of infected cells to the number of healthy cells). In the absence of spatial variation a typical transition to the endemic steady state is shown in Figure \ref{fig:averages_2d_uniform}.

\begin{figure}[h]
  \begin{center}
    \includegraphics[width=0.9\columnwidth]{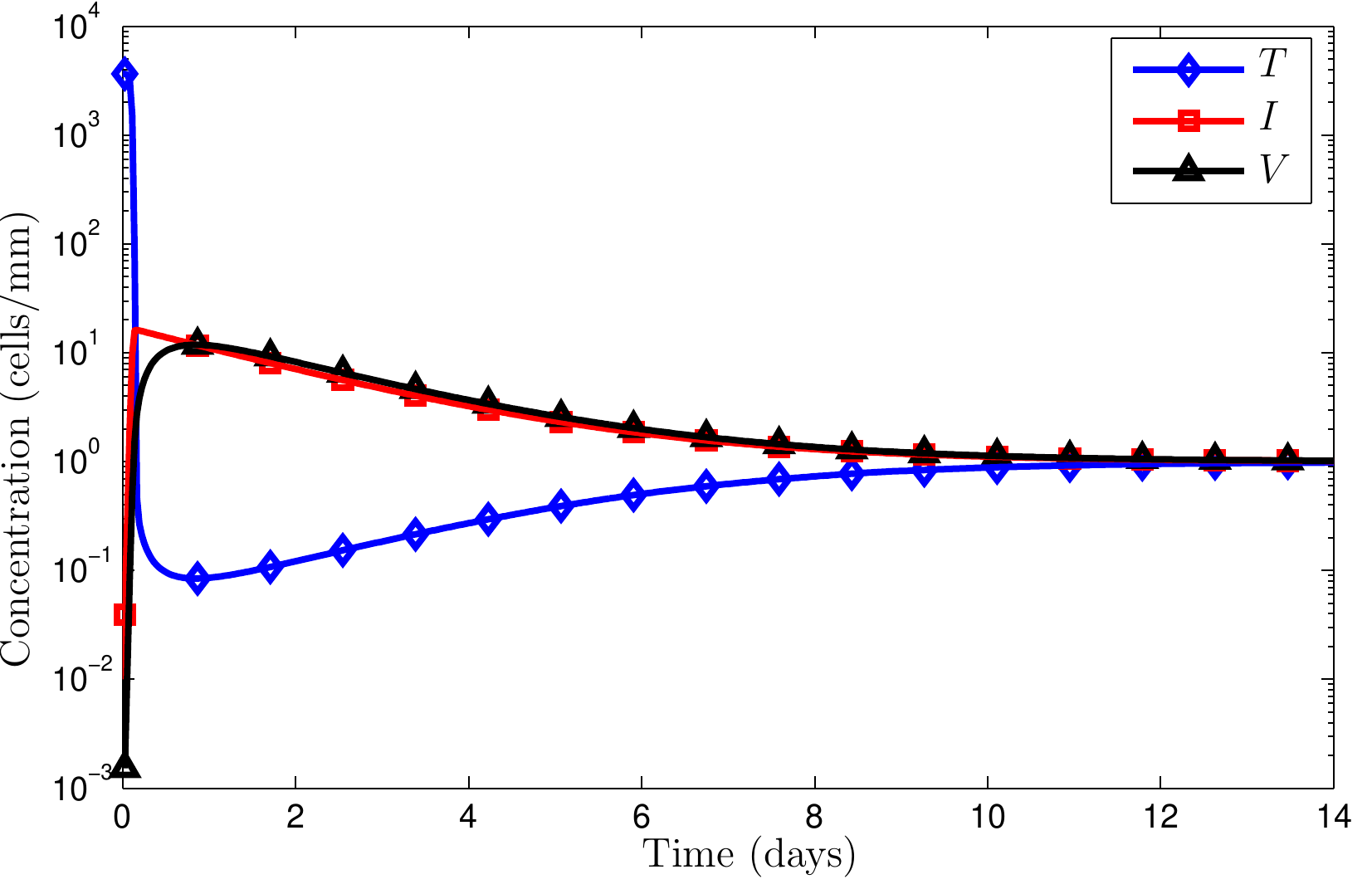}
  \end{center}
  \caption{A typical transition from disease-free steady state to endemic steady state in the absence of spatial variation.  }
  \label{fig:averages_2d_uniform}
\end{figure}

Figure \ref{fig:chemotaxis1d_chi130_trans} shows the transition from a small random disturbance around the disease-free steady state to the endemic steady state with chemotaxis $\chi=130\mathrm{\;mm^3\;cell^{-1}\;day^{-1}}$, above the critical threshold, regularised with $\epsilon = 1.0$. The resulting Turing pattern is identical to the pattern produced starting from a perturbation of the endemic steady state (see e.g. Figure \ref{fig:chemotaxis1d_chi130}).

\begin{figure}[h]
  \begin{center}
    \includegraphics[width=1\columnwidth]{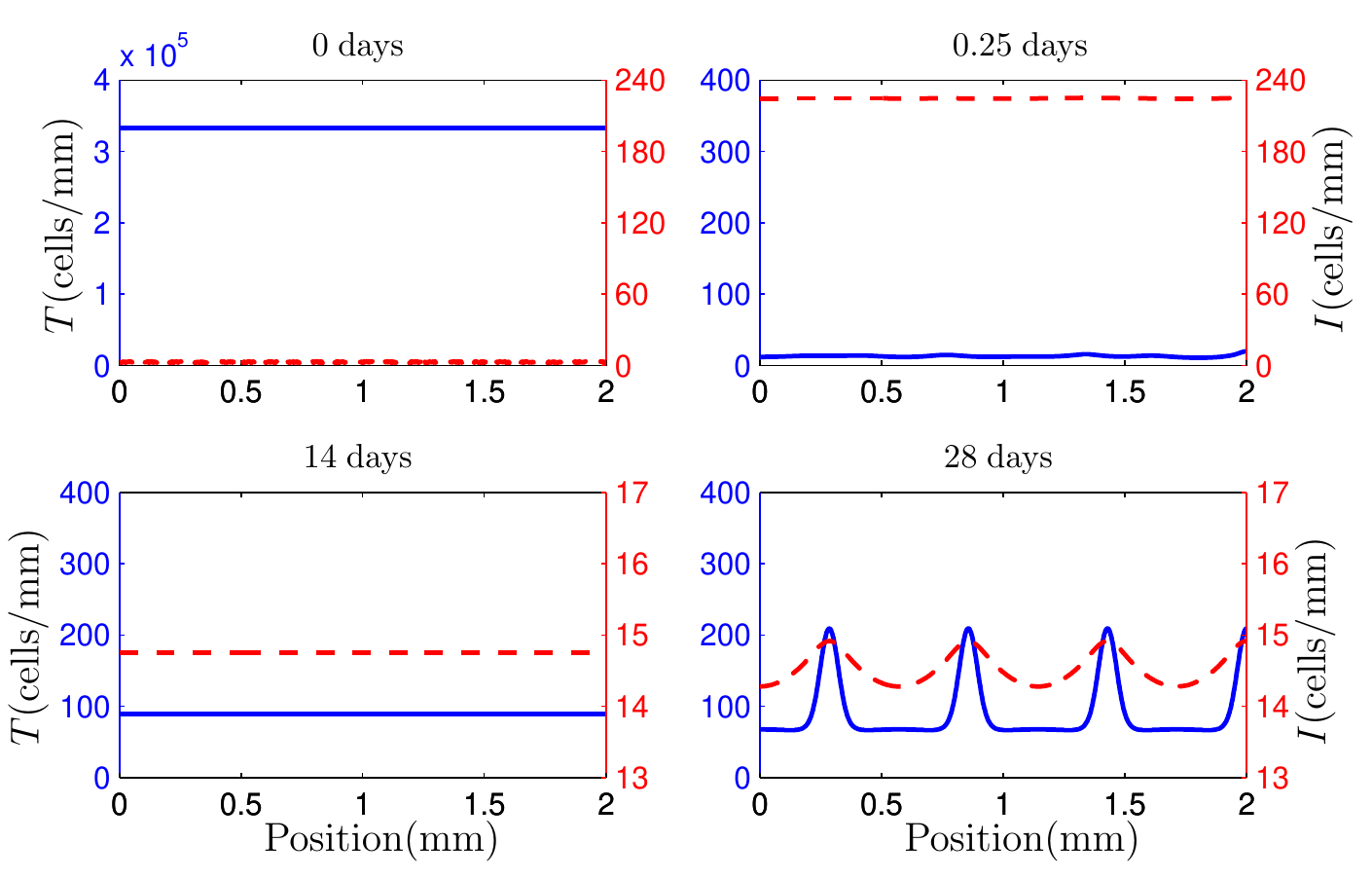}
  \end{center}
  \caption{A typical transition from the disease-free steady state to the endemic steady state in the presence of spatial variation. The chemotaxis parameters are $\chi = 130 \mathrm{\;mm^3\;cell^{-1}\;day^{-1}}$ and $\epsilon = 1.0$. Target cells $T$ (blue solid lines), infected cells $I$ (red dashed lines).}
  \label{fig:chemotaxis1d_chi130_trans}
\end{figure}

\begin{figure}[h]
  \begin{center}
    \includegraphics[width=1\columnwidth]{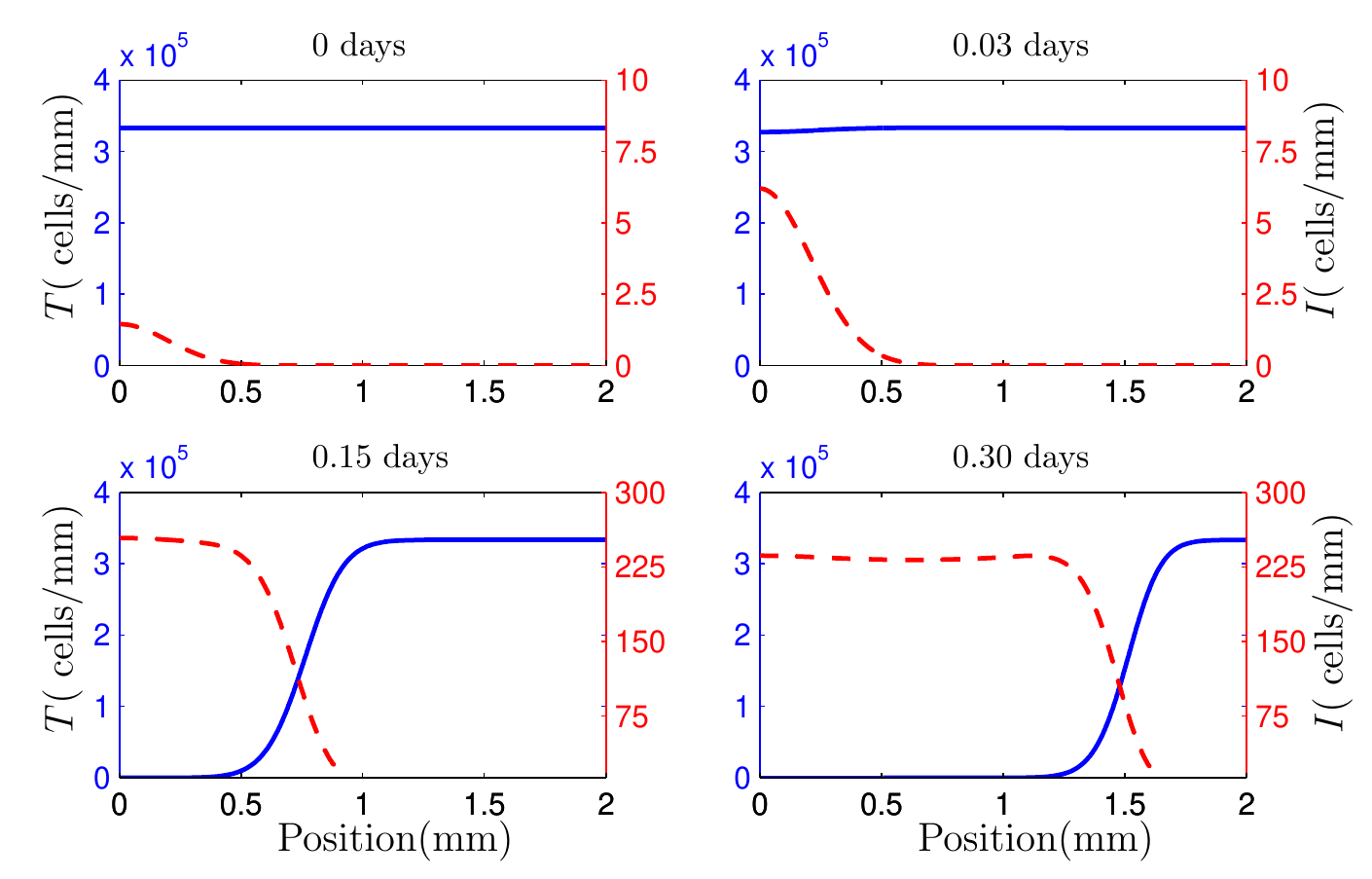}
  \end{center}
  \caption{Initial stage of infection with the distribution of infected cells localised near $X_1 = 0$. Chemotaxis was regularised with $\chi = 130\mathrm{\;mm^{3}\;cell^{-1}\;day^{-1}}$ and $\epsilon = 1.0$. The infection spreads at the rate of approximately $5\mathrm{\;mm}$ per day. Target cells $T$ (blue solid lines), infected cells $I$ (red dashed lines).}
  \label{fig:chemotaxis1d_chi130_trans_det}
\end{figure}

\subsection{Consequences of changing the virus clearance rate}
Figure \ref{fig:chemotaxis1d_chi130_changec} shows the change in the Turing pattern after the rate of virus clearance increases from $c=3\mathrm{\;day^{-1}}$ to $c=4\mathrm{\;day^{-1}}$. Even though the overall number of infected cell reduces, the maximum density of infected cells remains roughly the same as with the lower rate of clearance. Note that increasing  $c$ in this way will result in a lower chemotactic threshold for Turing pattern formation.

\begin{figure}[h]
  \begin{center}
    \includegraphics[width=1\columnwidth]{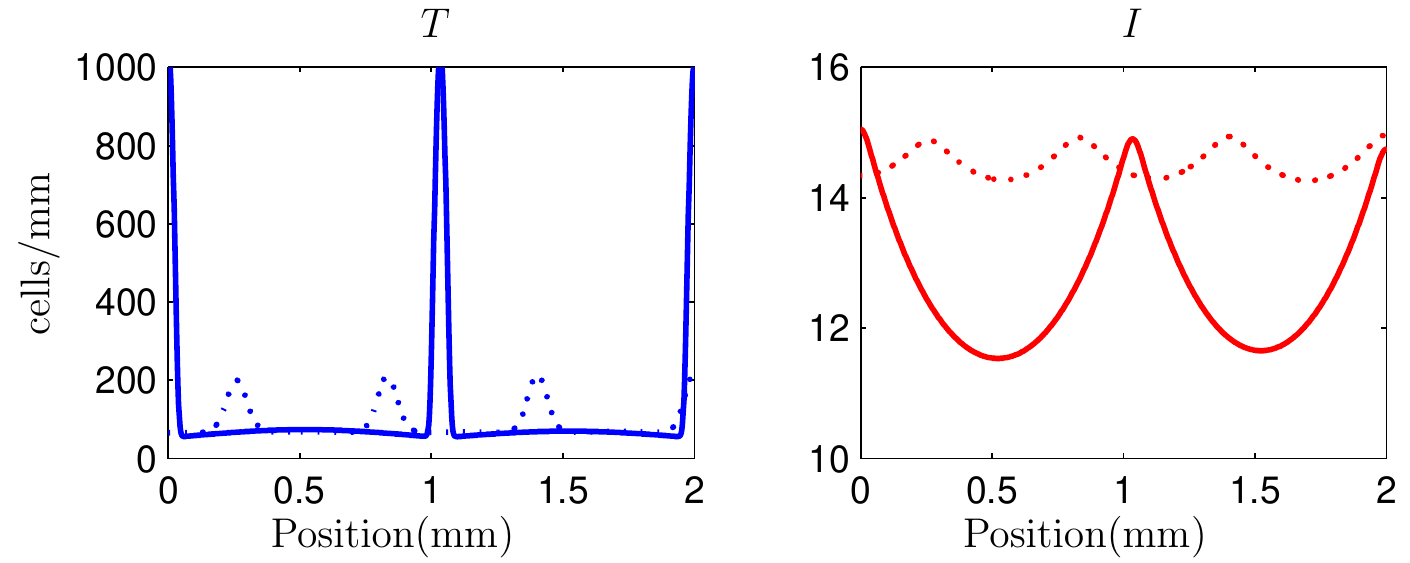}
  \end{center}
  \caption{Change in Turing patterns after clearance rate changes from $c=3\mathrm{\;day^{-1}}$ to $c=4\mathrm{\;day^{-1}}$. The dotted lines indicate the initial densities and the solid lines indicate the final densities of non-infected cells $(T)$ and infected cells $(I)$.}
  \label{fig:chemotaxis1d_chi130_changec}
\end{figure}

\section{Numerical Solutions in two spatial dimensions}

The two-dimensional model exhibits a blow up of solutions in finite time for sufficiently strong chemotactic attractions. For this reason we have employed a regularisation as described in \eqref{eq:model_reg}. Setting the thickness of the surface on which we are modelling the infection to $h=0.1\mathrm{mm}$, results in the critical value for chemotaxis $\chi$ to lie somewhere between $11$ and $12\;\mathrm{mm^4\;cell^{-1}\;day^{-1}}$.

The governing PDE was solved numerically on a square domain with sides of length $L\approx1.78\mathrm{\;mm}$ (i.e. $X_1 = X_2 = 1$). The method of lines was used, which involved semi-discretising \eqref{eq:model_reg} in both space variables on a uniform rectangular grid and solving the resulting ODE system in time using MATLAB's implementation of the Runge-Kutta method, {\tt ode45}. The results for the endemic steady state are shown in Figure \ref{fig:chemotaxis2d_chi12_eps10_rand} and Figure \ref{fig:chemotaxis2d_chi12_T}.

While the two-dimensional model is more difficult to solve numerically and some form of chemotactic regularisation is required we did not observe any qualitative behaviour that was different from the setting of only one dimension in space.
\begin{figure}[h!]
\begin{center}
\includegraphics[width=1\columnwidth]{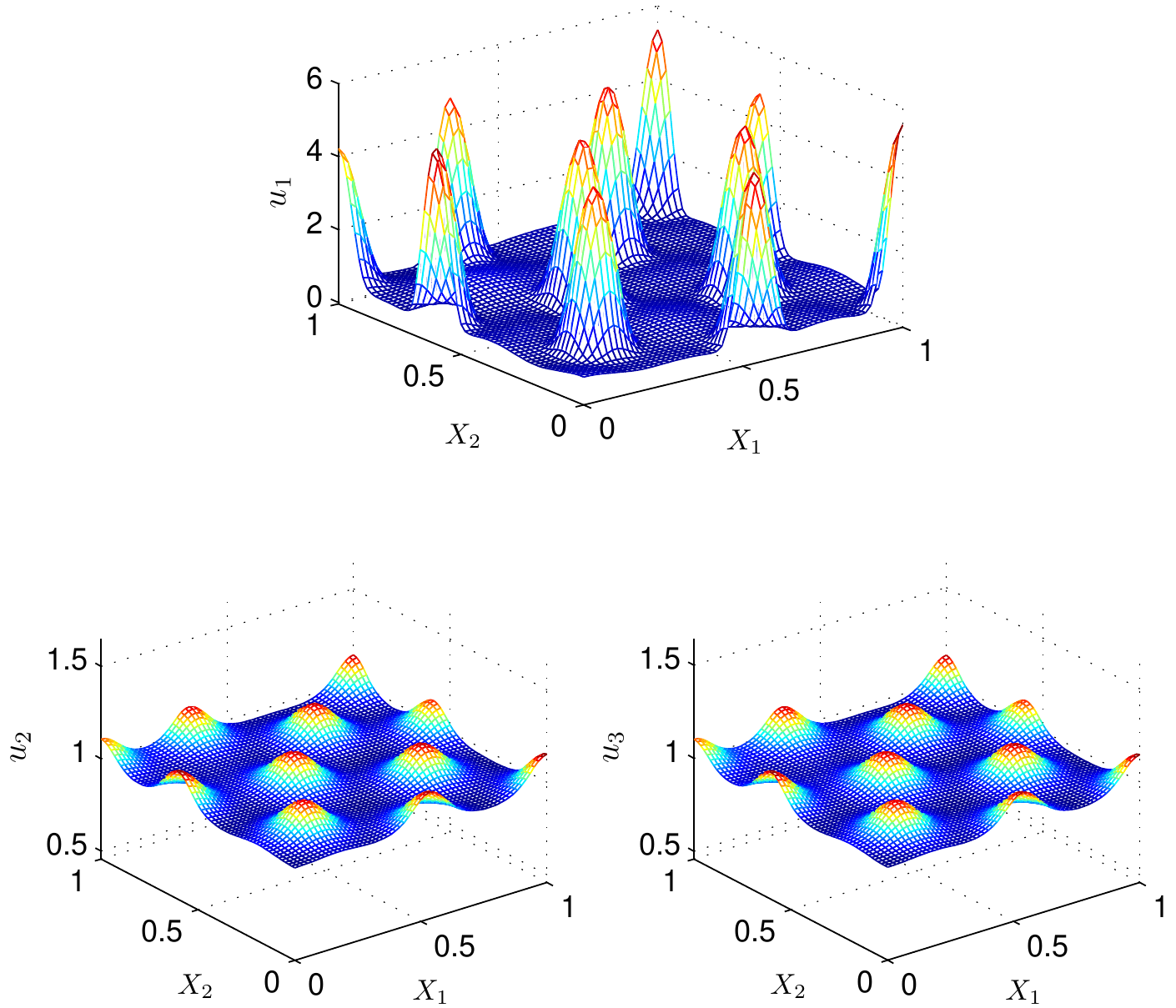}
\caption{Steady state solution for $\chi = 12\mathrm{\;mm^4\;cell^{-1}\;day^{-1}}$, $\epsilon = 10$.  }
\label{fig:chemotaxis2d_chi12_eps10_rand}
\end{center}
\end{figure}

\begin{figure}[h!]
\begin{center}
\includegraphics[width=1\columnwidth]{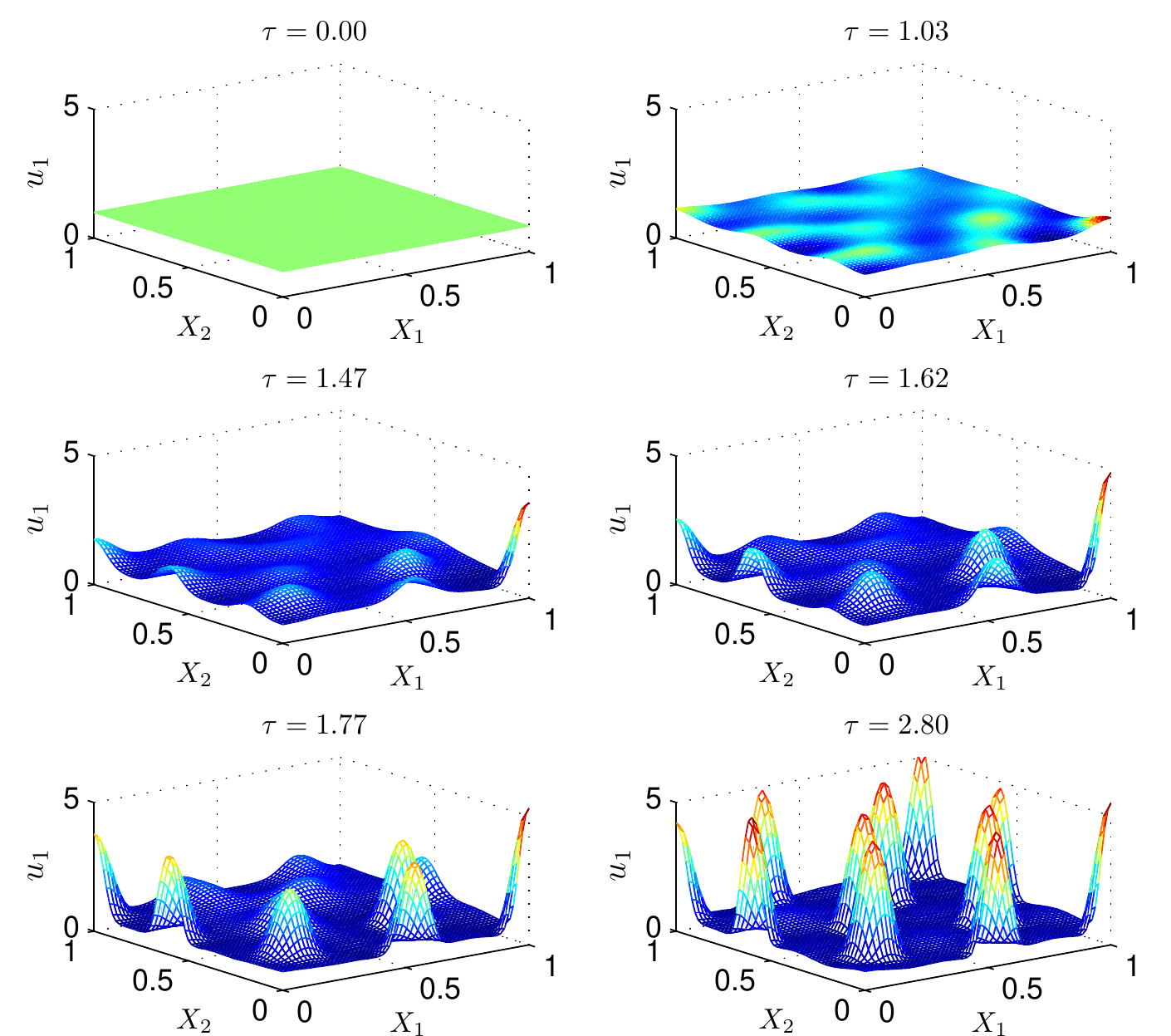}
\caption{Density of target cells at different points in time for $\chi = 12\mathrm{\;mm^4\;cell^{-1}\;day^{-1}}$, $\epsilon=10$. The initial conditions are taken as a random perturbation of the endemic steady state.}
\label{fig:chemotaxis2d_chi12_T}
\end{center}
\end{figure}

\section{Discussion}

We have shown analytically and numerically that HIV infection need not be spatially homogeneous and can exhibit Turing patterns. The patterns can only occur if i) the conditions for existence of an infected steady state hold ($skN>c\delta\mu$), and ii) the chemotactic attraction is strong enough. 
With the parameter values as outlined in Section 9, the estimated time for  transition to the nonhomogeneous endemic state to occur locally (i.e. on a small patch of tissue) is around 14 days (see Figure \ref{fig:averages_2d_uniform}).

Larger chemotactic attraction results in a larger range of possible spatial frequencies for the Turing patterns as well as higher amplitudes of the peaks. While the resulting pattern may depend on the initial state of the infection, we found that in most cases it does not. However, different initial states may significantly affect the speed at which the pattern is established, with small random perturbations of the endemic state enabling faster settling to a Turing pattern than localised perturbations which in turn settle faster than from a perturbation of the disease-free steady state.

We found that if the initial perturbation is local, the pattern tends to get established near the perturbation, and then propagates outwards, whereas for random initial perturbations the patterns emerge more or less simultaneously across the whole region.

Upon getting infected by HIV it typically takes 1-2 months \cite{Murray1998} for the body's immune system to respond by increasing the rate at which the virus is cleared. This is enough time for the initial Turing pattern to form. After the clearance rate is increased, it is possible for new, more prominent and less frequent Turing patterns to form. Also, even if the system was in a state that does not admit any patterns and the infection becomes spatially homogeneous, such a response from the immune system may give rise to pattern formation.

This analysis indicates that foci of HIV infection that are observed in tissue, can be established as a result of the dynamics of the system irrespective of any spatial heterogeneity in the tissue. These foci may provide an environment where new infection of cells outweighs immune clearance and hence supports the maintenance of infection despite an expanding immune response. Our analysis indicates that these patterns are established over a longer time scale than would be relevant for the impact of a microbicide which acts at the very earliest stages of infection. This analysis however assumes that the surface itself representing the vaginal or rectal epithelium, is homogeneous. This is not the case however and further work to assess the impact of the generation of Turing patterns in a more realistic environment is required.

\acknowledgement{We greatly acknowledge discussions with T.A.M. Langlands and P.J. Klasse on aspects of this work. This research was assisted through the support from the UNSW Goldstar Scheme.}
\clearpage

\appendix \label{appendix}
\section*{Appendix: Showing conditions for Turing instability}
\begin{proposition}
  Consider the characteristic polynomial \eqref{eq:chp_simple}. Provided $a_1, a_2, b_1, b_2, b_3, c_1, c_2, c_3$ and $c_4$ are all positive, conditions (S1)--(S4) hold.
\end{proposition}
\begin{proof} We demonstrate each of the conditions in turn below.
  \begin{itemize}
    \item[(S1)]
	\begin{align*}
	  a_1b_1 & = (\xi+\alpha+\beta)(\xi\alpha+\xi\beta) \\
	  & > \xi\alpha\beta\\
	  & > (\xi-1)\alpha\beta = c_1;
	\end{align*}
      \item[(S2)]
	\begin{align*}
  	a_1b_2 + a_2b_1 =& (\xi+\alpha)b_2 + \beta b_2 + (1+d_I)b_1 + d_V b_1\\
  	=& (\xi+\alpha)b_2 + (1+d_I)b_1 \\
	& + \beta(\alpha d_V + \beta d_I + \xi d_I + \alpha + \xi d_V + \beta - \alpha d_{\chi}) + \xi\alpha d_V + \xi\beta d_V\\
  	= & (\xi + \alpha)b_2 + (q+d_I)b_1 + \beta(\alpha d_V + \beta d_I + \alpha + \xi d_V + \beta) + \xi\beta d_V \\
  	& + (\xi\alpha d_V + \xi\beta d_I - \alpha \beta d_{\chi})\\
  	> & \xi\alpha d_V + \xi\beta d_I - \alpha \beta d_{\chi} =  c_2;
	\end{align*}
      \item[(S3)]
	\begin{align*}
	  a_1b_3  + a_2b_2 =& (\xi+\alpha+\beta)(d_Id_V + d_I + d_V) \\
	  &+(1+d_I+d_V)(\alpha d_V + \beta d_I + \xi d_I + \alpha + \xi d_V + \beta - \alpha d_\chi)\\
	   >& \xi d_I d_V + \alpha d_V + \beta d_I \\
	  & + (1+d_I)b_2 + d_V(b_2 + \alpha d_\chi) - \alpha d_\chi d_V\\
	  > & \xi d_I d_V + \alpha d_V + \beta d_I - \alpha d_\chi d_V = c_3;
	\end{align*}
      \item[(S4)]
	\begin{align*}
	  a_2b_3 = & (1+d_I + d_V)(d_I d_V + d_I + d_V)\\
	  > & d_I d_V = c_4.
	\end{align*}
  \end{itemize}
  \qed
\end{proof}

\bibliographystyle{plain}
\bibliography{hiv-2012}

\begin{thebibliography}{10}

\bibitem{CM2012}
S.~Barry Cooper and Philip~K. Maini.
\newblock The mathematics of nature at the alan turing centenary.
\newblock {\em Interface Focus}, 2(4):393--396, 2012.

\bibitem{DeBoer2007}
R.~J. De~Boer.
\newblock {Understanding the Failure of CD8+ T-Cell Vaccination against
  Simian/Human Immunodeficiency Virus}.
\newblock {\em J. Virol.}, 81(6):2838--2848, 2007.

\bibitem{Keshet2005}
L.~Edelstein-Keshet.
\newblock {\em {Mathematical Models in Biology}}, volume~46 of {\em Classics in
  Applied Mathematics}.
\newblock SIAM, 2005.

\bibitem{Haase2010}
A.~T. Haase.
\newblock {Targeting early infection to prevent HIV-1 mucosal transmission}.
\newblock {\em Nature}, 464:217--223, 2010.

\bibitem{Haase1996}
A.~T. Haase, K.~Henry, M.~Zupancic, G.~Sedgewick, R.~A. Faust, H.~Melroe,
  W.~Cavert, K.~Gebhard, K.~Staskus, Z.~Q. Zhang, P.~J. Dailey, H.~H. Balfour,
  A.~Erice, and A.~S. Perelson.
\newblock {Quantitative image analysis of HIV-1 infection in lymphoid tissue.}
\newblock {\em Science}, 274(5289):985--989, 1996.

\bibitem{Harris2012}
T.~H. Harris, E.~J. Banigan, D.~A. Christian, C.~Konradt, E.~D. Tait~Wojno,
  K.~Norose, E.~H. Wilson, B.~John, W.~Weninger, A.~D. Luster, A.~J. Liu, and
  C.~A. Hunter.
\newblock {Generalized Levy walks and the role of chemokines in migration of
  effector CD8+ T cells}.
\newblock {\em Nature}, 486:545--548, 2012.

\bibitem{Hillen2009}
T.~Hillen and K.~J. Painter.
\newblock {A user's guide to PDE models for chemotaxis}.
\newblock {\em J. Math. Biol.}, 58:183--217, 2009.

\bibitem{Horstmann2003}
D.~Horstmann.
\newblock {From 1970 until present: the Keller--Segel model in chemotaxis and
  its consequences}.
\newblock {\em I. Jahresberichte DMV}, 105(3):103--165, 2003.

\bibitem{Hurwitz1964}
A.~Hurwitz.
\newblock {On the conditions under which an equation has only roots with
  negative real parts}.
\newblock In {\em Selected Papers on Mathematical Trends in Control Theory},
  pages 72--82. Dover, New York, 1964.

\bibitem{Jin2008}
T.~Jin, X.~Xu, and D.~Hereld.
\newblock {Chemotaxis, chemokine receptors and human disease.}
\newblock {\em Cytokine}, 44(1):1--8, 2008.

\bibitem{Keele2008}
B.~F. Keele, E.~E. Giorgi, J.~F. Salazar-Gonzalez, J.~M. Decker, K.~T. Pham,
  M.~G. Salazar, C.~Sun, T.~Grayson, S.~Wang, H.~Li, X.~Wei, C.~Jiang, J.~L.
  Kirchherr, F.~Gao, J.~A. Anderson, L.-H. Ping, R.~Swanstrom, G.~D. Tomaras,
  W.~A. Blattner, P.~A. Goepfert, J.~M. Kilby, M.~S. Saag, E.~L. Delwart, M.~P.
  Busch, M.~S. Cohen, D.~C. Montefiori, B.~F. Haynes, B.~Gaschen, G.~S.
  Athreya, H.~Y. Lee, N.~Wood, C.~Seoighe, A.~S. Perelson, T.~Bhattacharya,
  B.~T. Korber, B.~H. Hahn, and G.~M. Shaw.
\newblock {Identification and characterization of transmitted and early founder
  virus envelopes in primary HIV-1 infection.}
\newblock {\em Proc. Nat. Acad. Sci. USA}, 105(21):7552--7557, 2008.

\bibitem{Klasse2006}
P.~J. Klasse, R.~J. Shattock, and J.~P. Moore.
\newblock {Which topical microbicides for blocking HIV-1 transmission will work
  in the real world?}
\newblock {\em PLoS Medicine}, 3(9):e351, 2006.

\bibitem{Lin2006}
F.~Lin and E.~C. Butcher.
\newblock {T cell chemotaxis in a simple microfluidic device.}
\newblock {\em Lab Chip}, 6(11):1462--1469, November 2006.

\bibitem{LJ2007}
Q.X. Liu and J.~Zhen.
\newblock Formation of spatial patterns in an epidemic model with constant
  removal rate of the infectives.
\newblock {\em J. Stat. Mech.: Theory and Experiment}, 2007(05):P05002, 2007.

\bibitem{MWBGL2012}
Philip~K. Maini, Thomas~E. Woolley, Ruth~E. Baker, Eamonn~A. Gaffney, and
  S.~Seirin Lee.
\newblock Turing's model for biological pattern formation and the robustness
  problem.
\newblock {\em Interface Focus}, 2(4):487--496, 2012.

\bibitem{Mandl2007}
J.~N. Mandl, R.~R. Regoes, D.~A. Garber, and M.~B. Feinberg.
\newblock {Estimating the Effectiveness of Simian Immunodeficiency
  Virus-Specific CD8+ T Cells from the Dynamics of Viral Immune Escape}.
\newblock {\em J. Virol.}, 81(21):11982--11991, 2007.

\bibitem{Meinhardt2012}
Hans Meinhardt.
\newblock Turing's theory of morphogenesis of 1952 and the subsequent discovery
  of the crucial role of local self-enhancement and long-range inhibition.
\newblock {\em Interface Focus}, 2(4):407--416, 2012.

\bibitem{Miller2003}
M.~J. Miller, S.~H. Wei, M.~D. Cahalan, and I.~Parker.
\newblock {Autonomous T cell trafficking examined in vivo with intravital
  two-photon microscopy.}
\newblock {\em Proc. Nat. Acad. Sci. USA}, 100(5):2604--2609, 2003.

\bibitem{MurrayBook}
J.~D. Murray.
\newblock {\em {Mathematical Biology: An Introduction}}, volume~19 of {\em
  Biomathematics}.
\newblock Springer, 2002.

\bibitem{Murray2007}
J.~M. Murray, S.~Emery, Kelleher~A. D., M.~Law, J.~Chen, D.~J. Hazuda, Nguyen
  B.-Y. T., H.~Teppler, and D.~A. Cooper.
\newblock {Antiretroviral therapy with the integrase inhibitor raltegravir
  alters decay kinetics of HIV, significantly reducing the second phase}.
\newblock {\em AIDS}, 21:2315--2321, 2007.

\bibitem{Murray1998}
J.~M. Murray, G.~Kaufmann, A.~D. Kelleher, and D.A. Cooper.
\newblock {A model of primary HIV-1 infection}.
\newblock {\em Math. Biosci.}, 154(2):57--85, 1998.

\bibitem{Murray2011}
J.~M. Murray, A.~D. Kelleher, and D.~A. Cooper.
\newblock {Timing of the Components of the HIV Life Cycle in Productively
  Infected CD4+ T Cells in a Population of HIV-Infected Individuals}.
\newblock {\em J. Virol.}, 85(20):10798--10805, 2011.

\bibitem{Nowak2000}
M.~A. Nowak and R.~May.
\newblock {\em {Virus dynamics: Mathematical principles of immunology and
  virology}}.
\newblock Oxford University Press, 2000.

\bibitem{Nowak1996}
M.A. Nowak, S.~Bonhoeffer, and A.M. Hill.
\newblock {Viral dynamics in hepatitis B virus infection}.
\newblock {\em Proc. Nat. Acad. Sci. USA}, 93(9):4398--4402, 1996.

\bibitem{OS1991}
Q.~Ouyang and Swinney~H. L.
\newblock Transition from a uniform state to hexagonal and striped turing
  patterns.
\newblock {\em Nature}, 352:610--612, 1991.

\bibitem{Perelson1997}
A.~S. Perelson, P.~Essunger, M~Cao, Y.~Vesanen, K.~Hurley, A.~Saksela,
  M.~Markowitz, and D.~D. Ho.
\newblock {Decay characteristics of HIV-1-infected compartments during
  combination therapy}.
\newblock {\em Nature}, 387:188--191, 1997.

\bibitem{Perelson1996}
A.~S. Perelson, Neumann~A. U., M.~Markowitz, J.~M. Leonard, and D.~D. Ho.
\newblock {HIV-1 Dynamics in Vivo: Virion Clearance Rate, Infected Cell
  Life-Span, and Viral Generation Time}.
\newblock {\em Science}, 271:1582--1586, 1996.

\bibitem{Ribeiro2010}
R.~M. Ribeiro, L.~Qin, L.~L. Chavez, D.~Li, S.~G. Self, and A.~S. Perelson.
\newblock {Estimation of the Initial Viral Growth Rate and Basic Reproductive
  Number during Acute HIV-1 Infection}.
\newblock {\em J. Virol.}, 84(12):6096--6102, 2010.

\bibitem{Turing1952}
A.~M. Turing.
\newblock {The Chemical Basis of Morphogenesis}.
\newblock {\em Philosophical Transactions of the Royal Society B: Biological
  Sciences}, 237(641):37--72, 1952.

\bibitem{Velazquez2004}
T.~J.~L. Velasquez.
\newblock {Point dynamics for a singular limit of the Keller--Segel model}.
\newblock {\em SIAM J. Appl. Math.}, 64(4):1198--1223, 2004.

\bibitem{Wei1995}
X.~Wei, S.~K. Ghosn, M.~E. Taylor, V.~A. Johnson, E.~A. Emini, P.~Deutsch,
  J.~D. Lifson, S.~Bonhoeffer, M.~A. Nowak, B.~H. Hahn, M.~S. Saag, and G.~M.
  Shaw.
\newblock {Viral dynamics in human immunodeficiency virus type 1 infection}.
\newblock {\em Nature}, 373:117--122, 1995.

\end{thebibliography}

\end{document}